\documentclass[conference,twoside]{IEEEtran}
\usepackage{siunitx}
\usepackage{graphicx}
\usepackage{amssymb}
\usepackage{amsmath}
\usepackage{algorithm}
\usepackage{algpseudocode}
\usepackage{cite}
\usepackage{stfloats}
\usepackage{epstopdf}
\usepackage{mdwlist}
\usepackage{threeparttable}
\usepackage[all,2cell]{xy} \UseAllTwocells
\usepackage{booktabs}
\usepackage{fancyvrb}
\usepackage{balance}
\usepackage{tensor}
\usepackage{leftidx}
\usepackage{mathabx}
\usepackage{color}
\usepackage{csquotes}
\usepackage{pdflscape}
\usepackage{afterpage}
\usepackage{array}
\usepackage{cases}
\usepackage{bm}

\newtheorem{theorem}{Theorem}
\newtheorem{lemma}{Lemma}
\newtheorem{corollary}{Corollary}
\newtheorem{proposition}{Proposition}

\newtheorem{remark}{Remark}

\AtBeginDocument{}

\newcommand{\erfc}{\mathrm{erfc}}

\newcommand{\diff}{\mathrm{d}}
\newcommand{\Pt}{P_\theta}

\newcommand{\atx}{a_\mathrm{tx}}
\newcommand{\arx}{a_\mathrm{rx}}

\newcommand{\Tx}{\mathrm{Tx}}
\newcommand{\Rx}{\mathrm{Rx}}
\newcommand{\TTx}{T_\mathrm{Tx }}
\newcommand{\TRx}{T_\mathrm{Rx }}
\newcommand{\Dt}{\Delta t}

\begin{document}
	
	\title{\huge
Diffusive Mobile MC for Controlled-Release Drug Delivery with Absorbing Receiver
	}
	
	\author{\IEEEauthorblockN{Trang Ngoc Cao\IEEEauthorrefmark{1}, Arman Ahmadzadeh\IEEEauthorrefmark{3}, Vahid Jamali\IEEEauthorrefmark{3}, Wayan Wicke\IEEEauthorrefmark{3}, \\ Phee Lep Yeoh\IEEEauthorrefmark{2}, Jamie Evans\IEEEauthorrefmark{1}, and Robert Schober\IEEEauthorrefmark{3}}
		\IEEEauthorblockA{\IEEEauthorrefmark{1}Department of Electrical and Electronic Engineering, The University of Melbourne, Australia}
		\IEEEauthorblockA{\IEEEauthorrefmark{2}School of Electrical and Information Engineering, The University of Sydney, Australia}
		\IEEEauthorblockA{\IEEEauthorrefmark{3}Institute for Digital Communications, Friedrich-Alexander-University Erlangen-Nurnberg (FAU), Germany}
	}

	\maketitle
	\begin{abstract}

		Nanoparticle drug carriers  play an important role in facilitating efficient  targeted drug delivery, i.e., improving  treatment success and reducing drug costs and side effects. However, the mobility of nanoparticle drug carriers poses a challenge in designing drug delivery systems. Moreover, healing results critically depend on the rate and time duration of drug absorption.
		Therefore, in this paper, we aim to design  a controlled-release drug delivery system with a mobile drug carrier that minimizes the total amount of released drugs while ensuring a desired rate of drug absorption   during a prescribed time period.
		We model the mobile drug carrier as a mobile transmitter, the targeted diseased cells as an absorbing receiver, and the channel between the transceivers  as  a time-variant  channel since  the carrier mobility results in a time-variant absorption rate of the drug molecules. Based on this, we   develop a molecular communication (MC) framework to design the controlled-release drug delivery system. 
		In particular, we  develop new analytical expressions for the mean, variance, probability density function, and cumulative distribution function of the channel impulse response (CIR). 
		Equipped with the statistical analysis of the CIR, we design and evaluate the performance of the controlled-release drug delivery system.
		Numerical results show significant savings in the amount of released drugs compared to a constant-release rate design and reveal the necessity of accounting for  drug carrier mobility for reliable  drug delivery.
	\end{abstract}
	\section{Introduction}
	

	
	In  drug delivery systems,  drug molecules are carried to the diseased cell site by  nanoparticle carriers, so that the drug is efficiently delivered to the targeted site and does not affect healthy  cells \cite{CMM:17:ST}. Experimental and theoretical studies have indicated that not only the total drug dosage but also the rate and time period of drug absorption by the diseased cell receptors are critical factors in the healing process \cite{LEELYP:15:CES,KS:14:DD}.  
	Therefore, it is important to design drug delivery systems with controlled release to minimize the total amount of released drugs while achieving a desired rate of drug absorption  at the diseased site during a prescribed time period. 
	
	To this end, the mobility of  drug carriers has to be accurately modeled due to the fact that  after being injected or extravasated from the cardiovascular system into the tissue surrounding a targeted diseased cell site, the drug carriers may not be anchored at the targeted site but may move, mostly via diffusion \cite{SSR:14:BE,PNJ:99:BPJ, LEELYP:15:CES,WCX:14:PCL}. The diffusion of the drug carriers results in a time-variant absorption rate of drugs  even if the  drug release rate is constant. 
	
	The challenge of designing a controlled-release drug delivery system has been tackled from two perspectives, namely mathematical modeling \cite{ARI:06:ADD} and experiments in vitro and vivo \cite{FWR:11:JP}. In particular, the mathematical approach helps explain the experimental observations and can guide the  experiments. Recently, researchers have started to design drug delivery systems based on the molecular communications (MC) paradigm where  drug carriers are modeled as transmitters,  diseased cells are modeled as  receivers, and   drug absorption is modeled as a random channel~\cite{FYE:16:CSTO}. Controlled-release designs based on an MC framework  were proposed  in \cite{CPA:15:BME,FRV:15:TNB,SMA:17:NCN,SMH:18:NB}. However, in these works, the transceivers were fixed and only the movement of  drug particles was considered. In contrast, in this paper, we account for the mobility of the transmitter and analyze the resulting time-variant MC channel  to optimize the controlled-release design. We note that time-variant MC channels with mobile transceivers were also considered in  \cite{AJS:18:COMT,HAA:17:CL,VHG:18:Arxiv}. In \cite{AJS:18:COMT}, a theory for stochastic time-variant channels in mobile diffusive MC systems was developed. However, a passive receiver model was used in \cite{AJS:18:COMT}, which may not be suitable for modeling drug delivery systems since the effect of drug absorption  cannot be captured.  A diffusive absorbing receiver and the average distribution of the first hitting time, i.e., the mean of the channel impulse response (CIR),  were derived for a one-dimensional environment without drift in \cite{HAA:17:CL} and with drift in \cite{VHG:18:Arxiv}. Clearly, none of these works provides a complete statistical analysis of
	the three-dimensional (3D) time-variant channel with an absorbing receiver
	nor do they consider drug delivery systems.

	In this paper, we analyze the 3D time-variant channel with  diffusive mobile transmitter and  absorbing receiver for a controlled-release drug delivery system. The main contributions are as follows:
	\begin{itemize}
		\item We design a controlled-release profile that  minimizes the amount of released drugs while ensuring that the absorption rate at the diseased cells  does not fall below a prescribed threshold for a given period of time. Our proposed design requires a significantly lower  amount of released drugs  compared to a design with  constant-release rate.
		\item We derive the first-order (mean) and second-order (variance) moments of the time-variant CIR and exploit them for the design of the controlled-release profile.
		\item We derive  the probability density function (PDF) and the cumulative distribution function (CDF) of the time-variant CIR for evaluation of the system performance.
		\item The performance of the controlled-release system is evaluated in terms of the probability that the absorption rate exceeds a targeted threshold.  Our results reveal that considering transmitter mobility is crucial for meeting the system requirements.
	\end{itemize}
	We note that whilst this paper focuses on  drug delivery, the derived analytical results  for the mean, variance, PDF, and CDF of the time-variant CIR  are expected to be also useful for other MC applications.	
	
	The remainder of this paper is organized as follows. In Section II, we introduce the system model. In Section III, we design  the controlled-release profile for a drug delivery system based on the mean and the variance of the time-variant CIR. In Section IV, we evaluate the performance of the controlled-release drug delivery system in terms of the probability that the absorption rate exceeds a target threshold. Numerical results are presented in Section V and  Section VI concludes the paper.

	\section{System and Channel Model}
	
	\begin{figure}[t!] 		
		\centering
		\includegraphics[width=0.45\textwidth]{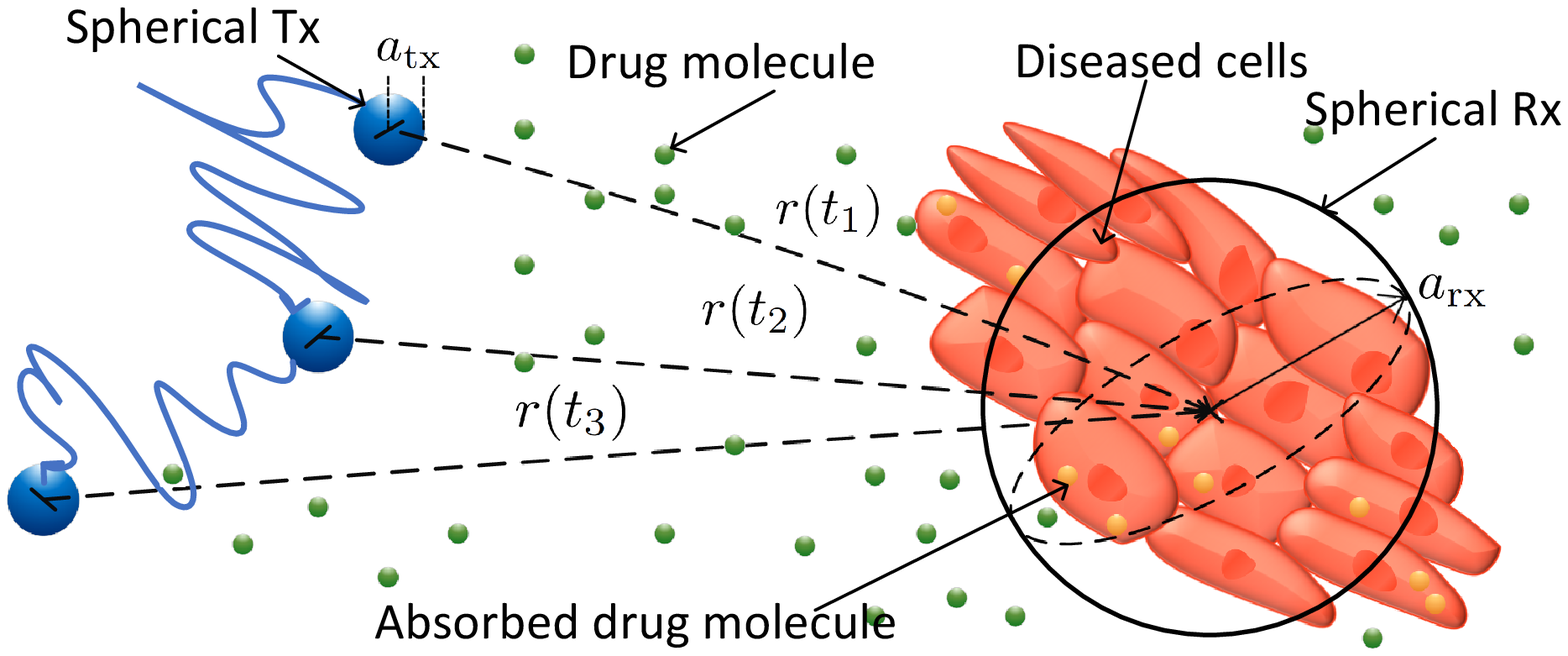}
		\caption{
			System model of controlled-release drug delivery with a diffusive transmitter ($\Tx$) and an absorbing receiver ($\Rx$).  
		}
		\label{fig:1}
	\end{figure}
	
	In this section, we introduce the diffusive mobile MC system model and define the time-variant CIR of the absorbing receiver.
	
	\subsection{System Model}
	
	We apply an MC design framework to model, analyze, and optimize a controlled-release drug delivery system, see Fig.~\ref{fig:1}.
	The drug delivery system comprises a drug carrier releasing drug molecules and diseased cells absorbing them. We model the system environment as an unbounded 3D diffusion environment with constant temperature and viscosity. 
	The drug carriers in  drug delivery systems are typically  nanoparticles, such as  spherical polymers or polymer chains, having a size not smaller than $\SI{100}{\nano\metre}$  \cite{PNJ:99:BPJ}. Hence, we  model the drug carrier as a spherical transmitter, denoted by $\Tx$, with radius $\atx$. 
	Furthermore, we model  the $\Tx$ as  transparent, i.e., it does not have any effect on the receiver or drug molecules after they are released from its center. This model is valid since in reality the drug carrier is designed to carry drug molecules and interaction with the drug or the receiver is not intended.  
	The drug carriers can be directly injected or extravasated from the blood to the interstitial tissue near the diseased cells (e.g. a tumor), where they start to move. The movement of the drug carrying nanoparticles in the tissue is caused by diffusion and convection mechanisms but diffusion is expected to be  dominant in  most cases \cite{SSR:14:BE,PNJ:99:BPJ, LEELYP:15:CES,WCX:14:PCL}. At the tumor site, the drug carrier releases drug molecules of type $\mathrm{X}$, which also diffuse in the tissue \cite{LEELYP:15:CES}. Hence, we adopt Brownian motion to model the diffusion of the  $\Tx$ and  the molecules  $\mathrm{X}$ \cite{CMM:17:ST}. 
	When the drug molecules hit the tumor, they are absorbed by receptors on the surface of the diseased cells  \cite{LEELYP:15:CES,KS:14:DD}. For convenience, we model the tumor as a spherical absorbing receiver, denoted by $\Rx$. In reality, the colony of cancer cells may potentially  have a different geometry, of course. However, as an abstract approximation, we model the cancer cells as one effective spherical receiver with radius $\arx$ and with a surface area equivalent to the total surface area of the tumor, i.e., the absorption on both the actual and the modeled surfaces is expected to be comparable \cite{SSR:14:BE}. 
	
	The absorption rate ultimately determines the  therapeutic impact of   the drug \cite{LEELYP:15:CES,KS:14:DD}. Thus, we make achieving a desired absorption rate  the objective for the system design. We will formally define the absorption rate in the next subsection but before doing so, we  define the   parameters and   assumptions used in the system model.
	We assume that the diffusion of the $\Tx$ and molecules $\mathrm{X}$ is independent of each other with corresponding diffusion coefficients, $D_\Tx$ and $D_\mathrm{X}$, respectively. We denote the time-varying distance between the centers of the $\Tx$ and the $\Rx$  at time $t$ by $r(t)$. In a 3D Cartesian space, $r(t)$ can be represented as $r(t)=\left(r_x,r_y,r_z\right)$. Then, the distance between the centers of the transceivers at time zero is denoted by $r(t=0)=r_0=\left(r_{x0},r_{y0},r_{z0}\right)$. We assume  that the $\Tx$ can release  molecules during a period of time denoted by $\TTx$. After this period, the drug carrier may be removed by  blood circulation or run out of drugs. 
	We assume that the $\Tx$ releases molecules at its center  instantaneously and discretely over time. Let $t_i$ and $\Dt_i$ denote the time instant  of the $i$-th release and the duration of the interval between the $i$-th  and  the $(i+1)$-th release, respectively. We have $i\in \{1,\dots,I\}$, where $I$ is the total number of releases during $\TTx$. We note that a continuous release can  be approximated by letting $\Delta t_i \rightarrow 0$, i.e., $I \rightarrow \infty$. Furthermore, let $\alpha_i$ and $A=\sum_{i=1}^I \alpha_i$  denote the number of drug molecules released at time $t_i$ and the total amount of drugs released during $\TTx$, respectively.

	\subsection{Impulse Response of the Diffusive Channel} \label{sub3:2}
	
	To evaluate the drug absorption rate at the $\Rx$ given the drug release profile at the $\Tx$, we first need to derive the CIR. Let $h(t,\tau)$ denote the hitting rate, i.e., the absorption rate of a given molecule, $\tau$ [\si{\second}] after its release at time $t$ at the center of $\Tx$. Note that the distance between the centers of the $\Tx$ and the $\Rx$, i.e., $r(t)$, is a random variable and a function of $t$. Hence, $h(t,\tau)$ may be referred to as the CIR, which completely characterizes the time-variant channel.   In  $h(t,\tau)$, variable $t$ denotes the time instant of the release of the molecules at $\Tx$ while variable $\tau$ represents the time period between the release at the $\Tx$ and the absorption at the $\Rx$. 
	
	For a given $r(t)$, the CIR $h(t,\tau)$ is given by \cite{YHT:14:CL}
	\begin{align} \label{eq:6}
	h(t,\tau)=\frac{\arx}{\sqrt{4\pi D_\mathrm{X} \tau^3}}\left(1-\frac{\arx}{r(t)}\right)\exp\left(-\frac{\left(r(t)-\arx\right)^2}{4 D_\mathrm{X} \tau}\right),
	\end{align}
	for $\tau>0$, and $h(t,\tau)=0$, for $\tau\leq 0$. 
	From the definition of $h(t,\tau)$, for $\Delta \tau \rightarrow 0$, we can interpret $h(t,\tau)\Delta \tau$ as the probability of absorption of a molecule  by the $\Rx$ between times $\tau$ and $\tau+\Delta \tau$ after the release at time $t$. If  $\alpha_i$ molecules are released at the $\Tx$ at time $t_i$, the expected number of molecules absorbed at the $\Rx$ between times $t$ and $t+\Delta t$, for $\Delta t \rightarrow 0$,  is equal to $\alpha_i h( t_i,t-t_i) \Delta t$, for $\tau=t-t_i$.  During the period $[0,t]$,   a total amount of  $A_t=\sum_i \alpha_i, \forall i | t_i < t,$ of drugs are released and thus, an expected total amount of  $y(t)=\sum_i\alpha_i h( t_i,t-t_i) \Delta t, \forall i | t_i < t$, of drugs are absorbed between times $t$ and $t+\Delta t$, for $\Delta t \rightarrow 0$. Let $g(t)$ denote the absorption rate of  molecules $\mathrm{X}$ at the $\Rx$  at time $t$, i.e., $g(t) =y(t)/\Delta t$, $\Delta t \rightarrow 0$.  Then, we have
	\begin{align}\label{eq:5a}
	g\left(t\right)=\underset{\forall i | t_i < t}{\sum}\alpha_i h\left(t_i,t-t_i\right).
	\end{align}
	As mentioned before, the absorption rate $g(t)$   at the tumor directly affects the healing efficacy of the drug. Hence, we will design the drug delivery system such that $g(t)$ does not fall below a prescribed value.

	\section{Controlled-Release Design} \label{sec:3}
	
	We first formulate the controlled-release design problem and then derive the mean and variance of the stochastic time-variant channel to solve the problem.

	\subsection{Problem Formulation}
	
	The treatment of many diseases requires the diseased cells to absorb a minimum rate of drugs during a given period of time \cite{KS:14:DD}. To design an efficient drug delivery system satisfying this requirement, we optimize the amounts of released drugs $\alpha_i$ such that   the total amount of released drugs, $A=\sum_{i=1}^I \alpha_i$, is minimized and the absorption rate $g(t)$ is  equal to or larger than a targeted rate, $\theta(t)$, for a period of time, denoted  by $\TRx$. Depending on the properties of the tumor, $\theta(t)$ may vary with time.
	Since $g(t)$ is a random variable, 
	we will design the system based on the first and second order moments of the CIR. In particular, we minimize $A=\sum_{i=1}^I \alpha_i$ subject to the constraint that the mean of $g(t)$ minus a certain deviation  is equal to or above a threshold during $\TRx$, i.e., $\mathrm{E}\left\{g(t)\right\}-\beta\Gamma\left\{g(t)\right\}\geq \theta(t)$ for $0\leq t \leq \TRx$, where $\mathrm{E}\left\{\cdot\right\}$ and $\Gamma\{\cdot\}$ denote  expectation and   standard deviation, respectively, and $\beta$ is a coefficient determining how much deviation from the mean is taken into account. Based on \eqref{eq:5a}, the constraint can be written as a function of $\alpha_i$ as follows
	\begin{align} \label{eq:35}
	&\mathrm{E}\left\{g(t)\right\}-\beta\Gamma\left\{g(t)\right\}\overset{(a)}{\geq} \\
	&\sum_{\forall i | t_i<t}\alpha_i\left(\mathrm{E}\left\{h\left(t_i,t-t_i\right)\right\}-\beta\Gamma\left\{h\left(t_i,t-t_i\right)\right\} \right)\geq \theta(t),\nonumber
	\end{align}
	where $0\leq t\leq \TRx$. Inequality $(a)$ in \eqref{eq:35} is due to Minkowski's inequality \cite{Ste:08:Lec}.
	Note that we may not be able to  find $\alpha_i$ such that \eqref{eq:35}  holds for all values of $\beta$ and $\theta(t)$. However, when $\mathrm{E}\left\{h\left(t_i,t-t_i\right)\right\}>\beta\Gamma\left\{h\left(t_i,t-t_i\right)\right\}$, i.e., either $\beta$  or $\Gamma\left\{h\left(t_i,t-t_i\right)\right\}$ is small, so that $\beta\Gamma\left\{h\left(t_i,t-t_i\right)\right\}$ is sufficiently small, we can always find $\alpha_i$  so that \eqref{eq:35}  holds for arbitrary  $\theta(t)$. Since time $t$ is a continuous variable, the constraint in \eqref{eq:35} has to be satisfied for all values of $t$, $0\leq t \leq \TRx$, and thus there is an infinite number  of constraints, each of which corresponds to one value of $t$. Therefore, we simplify the problem by relaxing the constraints to hold only for  a finite number of time instants $t=t_n=n\Delta t_n$, where $n=1,\dots,N$ and $\Delta t_n=\TRx/N$.
	Then,  the optimization problem for the design of $\alpha_i$ can be formulated as 
	\begin{subequations}\label{eq:36}
		\begin{align}
		\underset{\alpha_i\geq 0, \forall i}{\min} &A=\sum_{i=1}^I \alpha_i \label{eq:36a}\\
		\text{s.t. }  &\sum_{i,t_i<t}\alpha_i\left( m\left(t_i,n\Delta t_n-t_i\right)- \beta  \sigma\left(t_i,n\Delta t_n-t_i\right)\right)\nonumber\\
		&\hspace{1mm}\geq \theta(n\Delta t_n-t_i), \text{ for } n=1,\dots, N, \label{eq:36c}
		\end{align}
	\end{subequations}
	where $m\left(t,\tau\right)$  and   $\sigma\left(t,\tau\right)$ are the mean and the standard deviation of $h\left(t,\tau\right)$, respectively.
	In order to solve \eqref{eq:36}, we need to derive analytical expressions for $m\left(t,\tau\right)$  and   $\sigma\left(t,\tau\right)$. Moreover, since  $h\left(t,\tau\right)$ is a function of $r(t)$, which is a random variable, we first need to derive the distribution of $r(t)$ before  deriving  $m\left(t,\tau\right)$  and   $\sigma\left(t,\tau\right)$. Having $m\left(t,\tau\right)$  and   $\sigma\left(t,\tau\right)$ and treating the $\alpha_i$ as real numbers, \eqref{eq:36} can be readily solved as a linear program using existing algorithms or numerical software such as \textsc{Matlab}. We note that although the numbers of molecules $\alpha_i$ are integers, for tractability,  we  solve \eqref{eq:36} for real  $\alpha_i$ and quantize  the results to the nearest integer values.

	Note that the problem in \eqref{eq:36} is statistical in nature and provides instructive guidance for the offline design of the system.

	\subsection{Distribution of the $\Tx$-$\Rx$ Distance in a Diffusive System}
	
	In this subsection, we derive the distribution of $r(t)$. If the diffusion of  $\Tx$  follows Brownian motion in the entire 3D environment, we have 
	$r_x \sim \mathcal{N}\left(r_{x0}, 2 D_\Tx t\right), 
	r_y \sim \mathcal{N}\left(r_{y0}, 2 D_\Tx t\right), 
	r_z \sim \mathcal{N}\left(r_{z0}, 2 D_\Tx t\right)$.
	Then, $\frac{r(t)}{\sqrt{2 D_\Tx t}}$, denoted by $\gamma$, follows a noncentral chi distribution, denoted by $\mathcal{X}_k(\lambda)$, \cite{Rob:69:Bel}
	\begin{align}\label{eq:4}
	\gamma=\frac{r(t)}{\sqrt{2 D_\Tx t}}=\sqrt{\frac{r_{x}^2+r_{y}^2+r_{z}^2}{2 D_\Tx t}} \sim \mathcal{X}_k(\lambda),
	\end{align}
	with parameters $k=3$ and $\lambda=\sqrt{\frac{r_{x0}^2+r_{y0}^2+r_{z0}^2}{2 D_\Tx t}} =\frac{r_0}{\sqrt{2 D_\Tx t}}$. Thus, we can obtain the PDF of $r$ as follows
	\begin{align} \label{eq:5}
	f_{r(t)}(r)\overset{(a)}{=}&\frac{1}{\sqrt{2 D_\Tx t}}f_{\gamma}\left(\gamma\right)\\
	\overset{(b)}{=}&\frac{r}{r_0  \sqrt{\pi D_\Tx t}}\exp\left(-\frac{r^2+r_0^2}{4 D_\Tx t}\right)\sinh\left(\frac{r_0 r}{2 D_\Tx t}\right),\nonumber
	\end{align}
	where $f_{\gamma}\left(\gamma\right)$ is the PDF of  $\gamma$. Equality (a) in \eqref{eq:5} exploits the fact that $\gamma$ is  a function of $r(t)$ \cite[Eq.~5-16]{PP:02:Book}. Equality (b) in \eqref{eq:5} is obtained from the expression of the PDF $f_{\gamma}\left(\gamma\right)$ \cite{Rob:69:Bel}.
	
	\begin{remark}\label{re:1}
		Note that \eqref{eq:5} was derived under the assumption that the $\Tx$ can diffuse in the entire 3D environment. In reality, the $\Tx$ cannot be inside the $\Rx$, i.e., it does not interact with the $\Rx$, and thus will be reflected when it hits the $\Rx$ boundary. Hence, the actual $f_{r(t)}(r)$ may differ from  \eqref{eq:5}, e.g. $f_{r(t)}(r)=0$ for $r<\atx+\arx$. However,  we note that  for very small $r$, i.e., $r\approx 0$, \eqref{eq:5} does approach zero. Hence,   \eqref{eq:5} is a valid approximation for the actual $f_{r(t)}(r)$. The validity of this approximation is evaluated in Section~\ref{sec:5} via simulations.
	\end{remark}

	\subsection{Statistical Moments of Diffusive Channel} \label{sub3:3}
	
	In this subsection, we derive the statistical moments of the diffusive channel, i.e., $m(t,\tau)$ and $\sigma^2(t,\tau)$. In particular, $m(t,\tau)$ is obtained as
	\begin{align} \label{eq:7}
	m(t,\tau)=\int_0^\infty h(t,\tau)\left|_{r(t)=r}\right.f_{r(t)}(r)\diff r.
	\end{align}
	A closed-form expression of \eqref{eq:7} is provided in the following theorem.
	
	\begin{theorem}
		The mean of the impulse response of a time-variant MC channel  with diffusive molecules transmitted by a diffusive transparent transmitter and absorbed by an absorbing receiver is given by
		\begin{align} \label{eq:8}
		m(t,\tau)&= \frac{\arx}{4\sqrt{\pi\left(D_\mathrm{X} \tau+D_\Tx t\right)} r_0 \tau}\exp\left(-\frac{\arx^2}{4D_\mathrm{X} \tau}-\frac{r_0^2}{4D_\Tx t}\right)\nonumber \\
		&\times\left(-e^{\frac{b(t,\tau)^2}{4a(t,\tau)}}\left(\frac{b(t,\tau)}{2a(t,\tau)}+\arx\right)\erfc\left(\frac{b(t,\tau)}{2\sqrt{a(t,\tau)}}\right)\right.\nonumber\\
		&\left.+e^{\frac{c(t,\tau)^2}{4a(t,\tau)}}\left(\frac{c(t,\tau)}{2a(t,\tau)}+\arx\right)\erfc\left(\frac{c(t,\tau)}{2\sqrt{a(t,\tau)}}\right)\right),
		\end{align}
		where $\erfc(\cdot)$ is the complementary  error function and for compactness, $a(t,\tau)$, $b(t,\tau)$, and $c(t,\tau)$ are defined as follows
		\begin{align} \label{eq:9}
		a(t,\tau)&=\frac{1}{4D_\mathrm{X} \tau}+\frac{1}{4D_\Tx t},\hspace{5mm} b(t,\tau)=-\frac{\arx}{2D_\mathrm{X} \tau}-\frac{r_0}{2D_\Tx t} ,\nonumber\\ c(t,\tau)&=-\frac{\arx}{2D_\mathrm{X} \tau}+\frac{r_0}{2D_\Tx t}.
		\end{align}
	\end{theorem}
	
	\begin{proof}
		Substituting \eqref{eq:6} and \eqref{eq:5}  into \eqref{eq:7} and using the integrals given by \cite[Eq.~(2.3.15.4) and Eq.~(2.3.15.7)]{PBM:86:Book}, we obtain the expression for $m(t,\tau)$  in \eqref{eq:8}.	
	\end{proof} 
	
	\begin{remark}\label{re:2}
		We note that $m(t,\tau)$ approaches zero when $t~\rightarrow~\infty$ since $r(t)$ increases on average due to diffusion.
	\end{remark}
	
	Next, we obtain the variance of $h(t,\tau)$ as
	\begin{align} \label{eq:22}
	\sigma^2(t)=\phi(t)-m^2(t),
	\end{align}
	where $\phi(t)=\mathrm{E}\left\{h^2(t,\tau)\right\}$.
	The following lemma gives an analytical expression for  $\phi(t)$.
	\begin{lemma}
		For the considered channel, $\phi(t)$ is given by
		\begin{align} \label{eq:18}
		\phi(t)&=\hat{k}(t)\int_0^\infty \Big[ \exp\left(-\hat{a}(t)r^2-\hat{b}(t)r\right)\\
		&-\exp\left(-\hat{a}(t)r^2-\hat{c}(t)r\right)\Big]\left(r-2\arx+\frac{\arx^2}{r}\right)\diff r, \nonumber
		\end{align}
		where 
		\begin{align} \label{eq:19}
		&\hat{k}(t)=\frac{\arx^2e^{-\frac{\arx^2}{2D_\mathrm{X}\tau}-\frac{r_0^2}{4D_\Tx t}}}{8D_\mathrm{X} \pi \tau^3 r_0\sqrt{\pi D_\Tx t}},\hspace{5mm} \hat{a}(t)=\frac{1}{2D_\mathrm{X} \tau}+\frac{1}{4D_\Tx t},\nonumber\\
		&\hat{b}(t)=-\frac{\arx}{D_\mathrm{X}\tau}-\frac{r_0}{2D_\Tx t},\hspace{5mm} \hat{c}(t)=-\frac{\arx}{D_\mathrm{X}\tau}+\frac{r_0}{2D_\Tx t}.
		\end{align}	
	\end{lemma}
	\begin{proof}
		Substituting \eqref{eq:6} and \eqref{eq:5} into the definition of $\phi(t)$ and simplifying the expression, we obtain \eqref{eq:18}.
	\end{proof}
	
	\begin{remark}
		The expression in \eqref{eq:18} comprises  integrals of the form $\int_0^\infty \exp\left(a x^2+b x\right)/x \;\diff x$, where $a$ and $b$ are constants, and cannot be obtained in closed form. However, these integrals can be evaluated numerically in a straight forward manner.
	\end{remark}

	\section{Performance Analysis}
	
	Since $g(t)$ is random, we cannot always guarantee $g(t)\geq \theta(t)$.
	Moreover, since $g(t)\geq \theta(t)$ is required for proper operation of the system, we  evaluate the system performance in terms of the probability that  $g(t)\geq \theta(t)$, denoted by $\Pt=\mathrm{Pr}\left\{g(t)\geq \theta(t)\right\}$. In this section, we first present a theoretical framework for evaluation of the system performance in terms of $\Pt$ expressed as a function of the PDF and CDF of the CIR, before finally deriving the PDF and CDF of the CIR.

	\subsection{System Performance}
	
	The probability $\Pt$ is given in the following theorem.
	\begin{theorem}
		The system performance  metric $\Pt$ can be expressed as
		\begin{align}\label{eq:38}
		\Pt=&1-f_{\alpha_1 h\left(t-t_1,t_1\right)}\left(\theta(t)\right)\ast \dots \ast f_{\alpha_{\check{i}-1} h\left(t-t_{\check{i}-1},t_{\check{i}-1}\right)}\left(\theta(t)\right)\nonumber\\
		&\ast  F_{\alpha_{\check{i}} h\left(t-t_{\check{i}},t_{\check{i}}\right)}\left(\theta(t)\right),
		\end{align}
		where $\ast$ denotes  convolution, $\check{i}=1,2,\dots$ satisfies $t_{\check{i}}\leq t$, and $f_{\{\cdot\}}$ and $F_{\{\cdot\}}$ denote the PDF and CDF of the random variable in the subscript, respectively. In \eqref{eq:38}, we define $f_{\alpha_{\check{i}} h\left(t-t_{\check{i}},t_{\check{i}}\right)}\left(\theta(t)\right)=1/\alpha_{\check{i}}\times f_{ h\left(t-t_{\check{i}},t_{\check{i}}\right)}\left(\theta(t)/\alpha_{\check{i}}\right)$ and $F_{\alpha_{\check{i}} h\left(t-t_{\check{i}},t_{\check{i}}\right)}\left(\theta(t)\right)=F_{h\left(t-t_{\check{i}},t_{\check{i}}\right)}\left(\theta(t)/\alpha_{\check{i}}\right)$.
	\end{theorem}
	\begin{proof}
		From the definition of the CDF, we have
		\begin{align}\label{eq:38a}
		&\Pt=1-F_{g(t)}\left\{\theta(t)\right\}=1-\int_0^{\theta(t)}f_{g(t)}(g)\diff g.
		\end{align}
		Due to the summation of independent random variables in \eqref{eq:5a}, i.e., independent releases at $t_i$, we have
		\begin{align}\label{eq:38b}
		f_{g(t)}(g)=\left(f_{\alpha_1 h\left(t-t_1,t_1\right)}\ast \dots \ast f_{\alpha_{\check{i}} h\left(t-t_{\check{i}},t_{\check{i}}\right)}\right)(g).
		\end{align}
		Substituting \eqref{eq:38b} into \eqref{eq:38a}, and using the integration property of the convolution, i.e.,
		\begin{align}\label{eq:39}
		&\int_0^{\theta(t)}\left(f_{\alpha_1 h\left(t-t_1,t_1\right)}\ast \dots \ast f_{\alpha_{\check{i}} h\left(t-t_{\check{i}},t_{\check{i}}\right)}\right)(g)\diff g\\ \nonumber
		&=f_{\alpha_1 h\left(t-t_1,t_1\right)}\left(\theta(t)\right)\ast \dots \ast \int_0^{\theta(t)}f_{\alpha_{\check{i}} h\left(t-t_{\check{i}},t_{\check{i}}\right)}(g)\diff g,
		\end{align}
		and using the definition of the CDF, we obtain \eqref{eq:38}. 
	\end{proof}
	According to \eqref{eq:38}, $\Pt$ can be evaluated  based on exact expressions for the PDF  and the CDF of $h(t,\tau)$, which will be derived in the next subsection. 
	
	Furthermore, we note that a minimum value of $\Pt$ can be guaranteed based on statistical moments of the CIR, without  knowledge of the  PDF and the CDF, as shown in the following proposition.
	\begin{proposition} \label{prop:1}
		A lower bound on $\Pt=\mathrm{Pr}\left\{g(t)\geq \theta(t)\right\}$ is given as follows
		\begin{align}\label{eq:37}
		&\Pt\geq 1-\frac{1}{\beta^2}.
		\end{align}
	\end{proposition}
	\begin{proof}
		By using \eqref{eq:35} and the Chebyshev inequality \cite{PP:02:Book}, we obtain
		\begin{align}\label{eq:37a}
		&\Pt\overset{(a)}{\geq} \mathrm{Pr}\Big\{\left|g(t)-\mathrm{E}\left\{g(t)\right\}\right|\leq \mathrm{E}\left\{g(t)\right\}-\theta(t)\Big\}\\\nonumber
		&\overset{(b)}{\geq} 1-\frac{\Gamma^2\left\{g(t)\right\}}{ \left(\mathrm{E}\left\{g(t)\right\}-\theta(t)\right)^2}\overset{(c)}{\geq}  1-\frac{1}{\beta^2},
		\end{align}
		where (a) can be obtained easily by expanding the absolute value on the right-hand side, (b) is due to the  Chebyshev inequality, and (c) is due to \eqref{eq:35}. This completes the proof.
	\end{proof}
	\begin{remark} \label{re:3}
		Proposition~\ref{prop:1} is not only useful  for evaluating the system performance, but also provides a guideline for the design of the controlled release of drugs. For example,  to ensure $\Pt \geq 0.75$, from \eqref{eq:37}, we obtain $\beta=2$. Note that a useful bound can only be obtained based on \eqref{eq:37} when $\beta>1$.
	\end{remark}

	\subsection{Distribution Functions of the CIR} \label{sub4:2}
	
	The PDF of the CIR is  given in the following theorem.
	
	\begin{theorem} \label{theo:3}
		The PDF of the impulse response of a time-variant  channel  with diffusive molecules transmitted by a diffusive transparent transmitter and absorbed by an  absorbing receiver  is given by
		\begin{align} \label{eq:25}
		\begin{cases}
		f_{h\left(r(t),\tau\right)}(h)=\frac{f_{r(t)}(r_1(h))}{h'(r_1(h),\tau)}-\frac{f_{r(t)}(r_2(h))}{h'(r_2(h),\tau)}, &\text{for } 0 \leq h < h^\star,\\
		f_{h\left(r(t),\tau\right)}(h)\rightarrow \infty, &\text{for } h=h^\star,
		\end{cases}
		\end{align}
		where $h\left(r(t),\tau\right)$ denotes $h\left(t,\tau\right)$ as a function of $r(t)$ and $\tau$, $f_{r(t)}(r)$ is given by \eqref{eq:5},   $r_1(h)$ and $r_2(h)$, $r_1(h)< r_2(h)$, are the solutions of the equation $h\left(r(t),\tau\right)=h$, $h^\star$ is the maximum value of $h\left(r(t),\tau\right)$ for all values of $r(t)$,
		and $h'(r,\tau) $ is given by
		\begin{align} \label{eq:27}
		h'(r,\tau)=&\frac{\arx}{\sqrt{4\pi D_\mathrm{X} \tau^3}}\exp\left(-\frac{\left(r-\arx\right)^2}{4 D_\mathrm{X} \tau}\right)\\ \nonumber
		&\times \left(\frac{\arx}{r^2}-\frac{\left(r-\arx\right)}{2D_\mathrm{X} \tau}\left(1-\frac{\arx}{r}\right)\right).
		\end{align}
	\end{theorem}
	
	\begin{proof}
		From \eqref{eq:27}, we observe that $h'(r,\tau)=\frac{\partial h(r,\tau)}{\partial r}=0$ is equivalent to  a cubic equation in $r$, given by  $a r^3+b r^2+ c r +d=0$, with properly defined coefficients $a$, $b$, $c$, $d$ and  discriminant $\Delta=18abcd-4b^3 d+b^2 c^2-4a c^3-27a^2 d^2<0$. Hence, $h'(r,\tau)=0$ has only one real valued solution, denoted by $r^\star$.  Then, from \eqref{eq:27}, we obtain that $h'(r_1,\tau)>0$ for $r_1<r^\star$, $h'(r_2,\tau)<0$ for $r_2>r^\star$, and $h'(r,\tau)=0$ for $h=h^\star$, where $h^\star=h\left(r^\star,\tau\right)$ is the maximum value of $h(r,\tau)$. Finally, we derive \eqref{eq:25} by exploiting \cite[Eq.~5-16]{PP:02:Book} for the PDF of  functions of  random variables.  
	\end{proof}

	The CDF of $h(t,\tau)$ is given in the following corollary.
	\begin{corollary} \label{theo:4}
		The CDF of the impulse response of a time-variant  channel  with diffusive molecules transmitted by a diffusive transparent transmitter and absorbed by an absorbing receiver  is given by
		\begin{align} \label{eq:28}
		F_{h\left(r(t),\tau\right)}(h)=F_{r(t)}(r_1(h))+1-F_{r(t)}(r_2(h)), 
		\end{align}
		for $0 \leq h \leq h^\star$, where $F_{r(t)}(r)$ is the CDF of $r$ and is given by
		\begin{align} \label{eq:29}
		F_{r(t)}(r)=1-\mathbf{Q}_{\frac{3}{2}}\left(\lambda,\frac{r}{\sqrt{2D_\Tx t}}\right).
		\end{align}
		Here, $\lambda$ is defined in \eqref{eq:4} and $\mathbf{Q}_{M}\left(a,b\right)$ is the  Marcum Q-function as defined in \cite{Rob:69:Bel}.
	\end{corollary}
	
	\begin{proof}
		From the definition of the CDF and \eqref{eq:25}, we have
		\begin{align} \label{eq:30}
		F_{h\left(r(t),\tau\right)}(h)=&\int_0^{h} f_{h\left(r(t),\tau\right)}(\check{h})\diff \check{h} \\ \nonumber
		=&\int_0^{h} \frac{f_{r(t)}(\check{r}_1(\check{h}))}{ \partial h(\check{r}_1,\tau)/\partial \check{r}_1}-\frac{f_{r(t)}(\check{r}_2(\check{h}))}{ \partial h(\check{r}_2,\tau)/\partial \check{r}_2}\diff\check{h}\\ \nonumber
		=&\int_0^{{r}_1(h)} f_{r(t)}(\check{r}_1)\diff\check{r}_1-\int_\infty^{{r}_2(h)}f_{r(t)}(\check{r}_2) \diff\check{r}_2\\ \nonumber
		=&F_{r(t)}(r_1(h))+1-F_{r(t)}(r_2(h)),
		\end{align}
		where $\check{r}_1$ and $\check{r}_2$, $\check{r}_1<\check{r}_2$, are the solutions of the equation $h\left(\check{r},\tau\right)=\check{h}$.
		Moreover, since $\frac{r(t)}{\sqrt{2D_\Tx t}}$ follows a noncentral chi distribution, we obtain \eqref{eq:29} as \cite{Rob:69:Bel}
		\begin{align} \label{eq:31}
		F_{r(t)}(r)=F_{\frac{r(t)}{\sqrt{2D_\Tx t}}}\left(\frac{r}{\sqrt{2D_\Tx t}}\right)=1-\mathbf{Q}_{\frac{3}{2}}\left(\lambda,\frac{r}{\sqrt{2D_\Tx t}}\right).
		\end{align}
		This completes the proof.
	\end{proof}

	We note that the analytical expressions for the PDF and CDF of $h(t,\tau)$ in Theorem~\ref{theo:3} and Corollary~\ref{theo:4}, respectively, are not in closed form. Therefore, the evaluation of the system performance in \eqref{eq:38} can be approximated by a discrete convolution which is easily evaluated numerically.
	
	\begin{remark}
		The results for the mean, variance, PDF, and CDF  of the CIR in Sections~\ref{sub3:3} and \ref{sub4:2} can also be applied for  applications where both the transmitter and the receiver undergo diffusion. In this case, we have to replace  $D_\mathrm{X}$ and $D_\Tx$ in the derived expressions by $D_1=D_\mathrm{X}+D_\Rx$ and $D_2=D_\Tx+D_\Rx$, respectively. $D_1$ and $D_2$ are the effective diffusion coefficients capturing the relative movements of the molecules $\mathrm{X}$ and the $\Rx$ and the relative movements of the $\Tx$ and the $\Rx$, respectively, see \cite{AJS:18:COMT}.
	\end{remark}  
	
	\section{Numerical Results} \label{sec:5}

	In this section, we provide  numerical results to evaluate the accuracy of the derived expressions and the efficiency of the proposed drug delivery system.  In the simulations, we use a particle-based simulation of Brownian motion, where the transmitter performs a random displacement in  discrete time steps of length $\Delta t^{\mathrm{st}}$ seconds. The random displacement of  the transmitter  in each step is modeled as a Gaussian random variable with zero mean and standard deviation  $\sqrt{2D_\Tx\Delta t^{\mathrm{st}}}$. Furthermore, in the simulations, we also take into account the reflection of the  $\Tx$ when the $\Tx$ hits the $\Rx$. Moreover, we adopt Monte-Carlo simulation   by averaging our results over a large number of independent realizations of the $\Tx$ movement.

	For all numerical results, we use the set of simulation parameters in Table~\ref{tab:1}, unless otherwise stated. The parameters in Table~\ref{tab:1} are chosen to match real system parameters, e.g. the diffusion constants $D_\mathrm{X}$ of drug molecules vary from $10^{-9}$ to $10^{-14}$ \si{\metre^2\per \second} \cite{WCX:14:PCL}, the drug carriers have  sizes $\geq \SI{100}{\nano\meter}$ \cite{PNJ:99:BPJ}, the size of tumor cells is on the order of $\si{\micro\metre}$, and the drug carriers can be injected or extravasated from the cardiovascular system in the tissue surrounding the targeted diseased cell site \cite{LEELYP:15:CES}, i.e., close to  the tumor cells. The dosing periods in drug delivery systems are on the order of days \cite{ARI:06:ADD}, i.e., $\SI{24}{\hour}$. For simplicity, we set $\Delta t_i=\TTx/I, \forall i \in \{1,\dots,I\}$, and $N=5$, and  the value of the required absorption rate is set to $\theta(t)=\SI{1}{\second^{-1}}$.   We choose $I$   relatively large to obtain small intervals $\Delta t_i$. 
	All simulation results are averaged over $10^5$ independent realizations of the environment.
	
	\begin{table}[t!]
		\caption{System parameters used for numerical results}
		\centering
		\label{tab:1}
		\begin{tabular}{c|c||c|c}
			\toprule
			Parameter	&	Value	&	Parameter	&	Value \\[-0.05cm]
			\midrule
			$D_\mathrm{X}$ [$\mathrm{m^2/s}$]	&	$8\times 10^{-11}$		&	$D_\Rx$ [$\mathrm{m^2/s}$]	&	$0$ 	\\
			$\atx$ [$\mathrm{m}$]	&	$1\times 10^{-7}$&	$\arx$ [$\mathrm{m}$]	&	$1\times 10^{-6}$		\\
			$\TTx$ [$\si{\hour}$]	&	$24$& $\TRx$ [$\si{\hour}$]	 &	$24$	\\
			$r_0$ [$\mathrm{m}$]	&	$10\times 10^{-6}$&	$N$	&	$5$\\
			$I$	&	$3000$ & $\theta(t) [\si{\second^{-1}}]$ & $1$\\
			\bottomrule
		\end{tabular}
	\end{table}

	\begin{figure}[t!]
		\centering
		\includegraphics[width=0.39\textwidth]{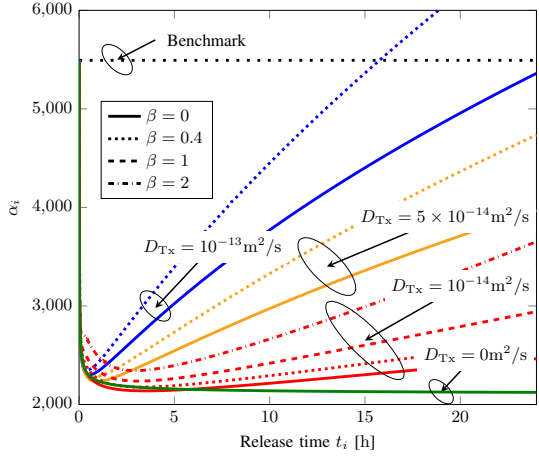}
		\caption{
			Optimal release coefficients $\alpha_i$ as a function of release time $t_i$ [\si{\hour}] for different system parameters. The black horizontal dotted line is the benchmark when the $\alpha_i$ are not optimized. 
		}
		\label{fig:4}
	\end{figure}

	In Fig.~\ref{fig:4}, we plot the controlled-release coefficients $\alpha_i$ versus the corresponding release time $t_i$ [$\si{\hour}$] for different system parameters. The  coefficients are obtained by solving the optimization problem in \eqref{eq:36} with $\beta=\{0,0.4,1,2\}$ for $D_\Tx=10^{-14}$ \si{\metre^2\per \second} and $\beta=\{0,0.4\}$ for $D_\Tx=\{5,10\}\times10^{-14}$ \si{\metre^2\per \second}. As mentioned  in the discussion of \eqref{eq:35}, we cannot  choose large values of $\beta$ when the diffusion coefficient is large, i.e., the standard deviation is large, as the problem may become infeasible. Fig.~\ref{fig:4} shows that for all considered parameter settings, we should first   release a large number of molecules for the absorption rate to exceed the threshold. Then, in the static system with  $D_\Tx= \SI{0}{\metre^2\per \second}$, the optimal coefficient  decreases with increasing time, since a fraction of the molecules previously released   from the $\Tx$ linger around the $\Rx$ and are absorbed later. However, for the mobile time-variant channels, the $\Tx$ eventually diffuses away from the $\Rx$ as time  $t$ increases and hence,  molecules released at later times by the $\Tx$  are far away from the $\Rx$ and may not reach the $\Rx$. Therefore, at later times, the  amount of drugs released has to be increased  for the absorption rate to not fall below the threshold. For higher $D_\Tx$, the $\Tx$ diffuses away from the $\Rx$ faster and thus, the coefficients $\alpha_i$ have to increase faster. This type of drug release is called a tri-phasic release \cite{FWR:11:JP}. Once we have designed the controlled-release profile, we can implement this by choosing a suitable drug carrier as shown in \cite{FWR:11:JP}.  Moreover, as expected, with larger $\beta$, we need to release more drugs to ensure  that \eqref{eq:36} is feasible. The black horizontal dotted line in Fig.~\ref{fig:4} is a benchmark where the $\alpha_i, \forall i,$ are not optimized but naively set to $\alpha_i=\alpha_1$. For this naive design, $A=\alpha_1 I\approx 1.65\times 10^7$, whereas with the optimal $\alpha_i$, for $\beta=0$ and $D_\Tx=10^{-13}$ [\si{\metre^2\per \second}], $A=1.2 \times 10^7$, i.e., equal to $73\%$ of  the naive design, and for $\beta=1$ and $D_\Tx=10^{-14}$ [\si{\metre^2\per \second}], $A=7.6 \times 10^6$, i.e., equal to $46\%$ of the naive design. This highlights that applying the optimal controlled-release profile can save significant amounts of drugs and still satisfy the therapeutic requirements.  Moreover, as observed in Fig.~\ref{fig:4}, the required values of $\alpha_i$ increase as $t_i$ increases and thus the naive design with  fixed $\alpha_i$, i.e., the benchmark, cannot ultimately satisfy the required absorption rate.
	
	\begin{figure}[t!]
		\centering
		\includegraphics[width=0.39\textwidth]{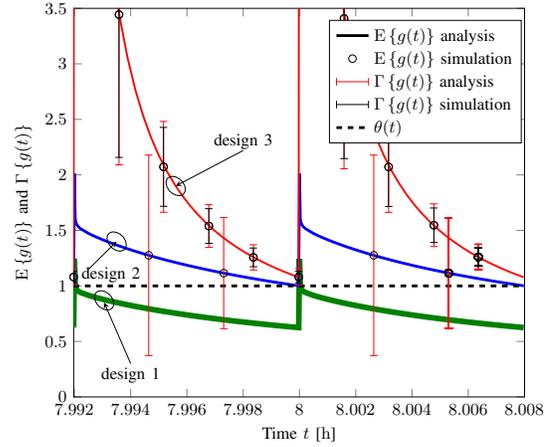}
		\caption{
			$\mathrm{E}\left\{g(t)\right\}$ and $\Gamma\left\{g(t)\right\}$ between the 1000-th release and the 1002-th release, i.e., at about \SI{8}{\hour}, for three different designs. Design 1 (green line): naive design without considering $\Tx$ movement  with $D_\Tx=10^{-13}$ \si{\metre^2\per \second} and $\beta=0$; design 2 (blue line) and 3 (red line): optimal design with $\left(D_\Tx[\si{\metre^2\per \second}],\beta\right)=\left(10^{-13},0\right),\left(10^{-14},1\right)$, respectively.
		}
		\label{fig:5}
	\end{figure}

	In Fig.~\ref{fig:5}, we plot the mean and standard deviation of the absorption rate,	$\mathrm{E}\left\{g(t)\right\}$ and $\Gamma\left\{g(t)\right\}$, between the  $1000$-th release and the $1002$-th release for three designs, where we adopted $D_\Tx=10^{-13}$ \si{\metre^2\per \second} and $\beta=0$ for designs 1 and 2, and $D_\Tx=10^{-14}$ \si{\metre^2\per \second} and $\beta=1$ for design 3. Note that the considered time window, e.g., between the  $1000$-th release and the $1002$-th release, is chosen arbitrarily in the middle of $\TTx$ to analyze the system behavior. For design 1, the $\Tx$ diffuses with $D_\Tx=10^{-13}$ \si{\metre^2\per \second} but the controlled release is  designed without accounting for the $\Tx$ mobility, i.e., the adopted $\alpha_i$ are given by the green line in Fig.~\ref{fig:4} obtained under the assumption of $D_\Tx= \SI{0}{\metre^2\per \second}$. For designs 2 and 3,  the mobility of the $\Tx$ is taken into account. The black dashed line marks the threshold $\theta(t)$ that $g(t)$ should not fall below. It is observed from Fig.~\ref{fig:5} that   when the $\Tx$ diffuses but the design does not take into account the mobility,	 the  requirement that the expected  absorption rate,  $\mathrm{E}\left\{g(t)\right\}$, exceeds $\theta(t)$, is not satisfied for most of the time. For design 2 with $\beta=0$, we observe that $\mathrm{E}\left\{g(t)\right\}>\theta(t)$ always holds but  $\mathrm{E}\left\{g(t)\right\}-\Gamma\left\{g(t)\right\}>\theta(t)$ does not always hold. For design 3 with $\beta=1$, we  observe  that $\mathrm{E}\left\{g(t)\right\}-\Gamma\left\{g(t)\right\}>\theta(t)$  always holds since $\beta>0$ enforces a gap between $\mathrm{E}\left\{g(t)\right\}$ and $\theta(t)$. In other words, even if $g(t)$ deviates from the mean, it can still exceed $\theta(t)$. Fig.~\ref{fig:5} also shows that $\mathrm{E}\left\{g(t)\right\}$  first increases after a release and then decreases, due to the diffusion of the molecules.   Furthermore, Fig.~\ref{fig:5}  confirms the accuracy of our derivations as  the simulation  results match the analytical results.  Note that in the  simulations, unlike the analysis, we have considered the reflection of the $\Tx$ when it hits the $\Rx$. Therefore, the good agreement in Fig.~\ref{fig:5} suggests that the reflection of the $\Tx$ does not have a significant impact on the numerical results and the approximation in \eqref{eq:5} is valid.
	
	In Fig.~\ref{fig:6}, we present the system performance in terms of $\Pt$,  for the time period between the $1000$-th and $1002$-th releases, i.e., at about $8 \si{\hour}$. The lines and markers denote  simulation and analytical results, respectively. Fig.~\ref{fig:6} shows a good agreement between the analytical and simulation results.  In Fig.~\ref{fig:6}, we observe that  $\Pt$ increases with increasing $\beta$ because the design for larger $\beta$ enforces a larger gap between $\mathrm{E}\left\{g(t)\right\}$ and $\theta(t)$, as can be seen in Fig.~\ref{fig:5}. Moreover, for a given $\beta$, $\Pt$ will be different for different $D_\Tx$. In particular, for larger $D_\Tx$, $\Pt$ is smaller due to the faster diffusion and less certainty about the CIR. Moreover, in Fig.~\ref{fig:6}, the green line shows that  the naive design, i.e., design 1 in Fig.~\ref{fig:5}, has very poor  performance. In  Fig.~\ref{fig:6}, we also observe that between two releases, $\Pt$ first increases due to the released drugs and then decreases due to drug diffusion.
	Furthermore, in Fig.~\ref{fig:6}, we show the lower bound on $\Pt$ derived in Proposition~\ref{prop:1} for $D_\Tx=10^{-14}$ \si{\metre^2\per \second} and $\beta=2$, where \eqref{eq:37} yields $\Pt\geq 0.75$. Fig.~\ref{fig:6} shows that the red dash-dotted line, i.e., $\Pt$ for $D_\Tx=10^{-14}$ \si{\metre^2\per \second} and $\beta=2$, is above the horizontal black dashed line, i.e., $\Pt=0.75$.
	
	\begin{figure}[t!]
		\centering
		\includegraphics[width=0.39\textwidth]{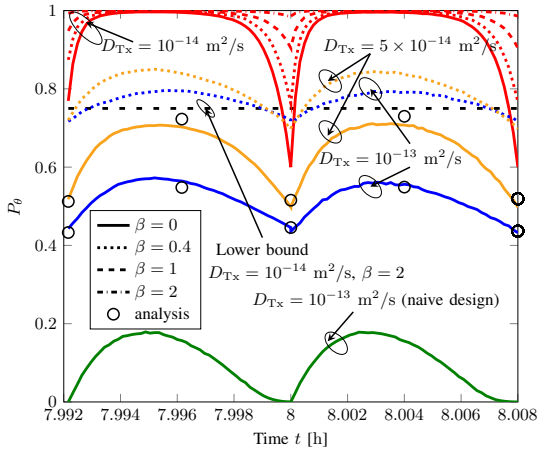}
		\caption{
			$\Pt$  as a function of time $t$ [$\si{\hour}$]  between 1000-th release and 1002-th release, i.e., at about $\SI{8}{\hour}$. 
		}
		\label{fig:6}
	\end{figure}
	
	\section{Conclusions}
	In this paper, we considered a drug delivery system with a diffusive drug carrier and  absorbing cells and modeled it as a time-variant channel between  diffusive MC transceivers. We provided a statistical analysis of the time-variant CIR. Based on this statistical analysis,  we designed the optimal controlled-release profile which minimizes the amount of released drugs while ensuring a targeted absorption rate of the drugs at the $\Rx$ for a prescribed time period. The probability of satisfying the constraint on the absorption rate was adopted as a system performance criterion and was evaluated. We observed that ignoring the reality of $\Tx$ mobility in designing the release profile leads to unsatisfactory performance.
	
	\bibliographystyle{IEEEtran}
	\bibliography{IEEEabrv,MC}

\end{document}